\documentclass[runningheads]{llncs}
\usepackage[left]{lineno} % 'switch' option works with two-column format
\usepackage[T1]{fontenc}

\usepackage{tikz,pgfplots}
\pgfplotsset{compat=1.18} 

\usepackage{graphicx}
\usepackage{hyperref}
\usepackage{amsmath}
\usepackage[english,rules]{ALgo}
\usepackage{float}

\usepackage{comment}

\usepackage{graphicx} % Required for inserting images

\newtheorem{obs}{Observation}

\def\per{\textit{per}}
\def\Border{\textit{Border}}
\def\border{\textit{border}}
\def\pp{\mathinner{\ldotp\ldotp}}	% two dots

\def\suflen{\textit{k}}
\def\period{\textit{period}}
\def\pos{\textit{pos}}
\def\listpos{\textit{listPos}}
\def\gap{\textit{gap}}
\def\mxbar{\bar{m}}
\def\cumsum{\dot{x}}

\def\sa#1{\mbox{\tt #1}}

\author{Thierry Lecroq \inst{1} \and Francesco Pio Marino\inst{1,2}}

\institute{
Univ Rouen Normandie, INSA Rouen Normandie, Universit\'e Le Havre Normandie, Normandie Univ, LITIS UR 4108, CNRS NormaSTIC FR 3638, IRIB,\\ Rouen F-76000, France\\
\and
Universit\`a di Catania, Dipartimento di Matematica e Informatica,\\ viale A.Doria n.6, 95125, Catania, Italia }

\authorrunning{T. Lecroq \and F.P. Marino}

\title{Efficient Computation of Periods and Covers Using Sampling}
\begin{document}
%\linenumbers 

\maketitle
\begin{abstract}
    Identifying regularities in strings, such as \emph{periods} and \emph{covers}, is crucial for applications in text compression, computational biology, and pattern recognition. \emph{Characters-Distance-Sampling} (\texttt{CDS}) is an efficient technique that encodes a string by storing distances between selected pivot characters, accelerating string-processing tasks.

    We apply \texttt{CDS} to compute periods and shortest covers, selecting only the first character as the pivot. This strategy yields optimized computations, achieving speedups of $38\%$--$43\%$ for period computation and $63\%$--$72\%$ for cover detection. These results demonstrate the potential of \texttt{CDS}-based representations for efficient string analysis and broader applications.
\end{abstract}

\section{\label{sect-intro}Introduction}

String analysis is fundamental in various domains, including literature, linguistics, computer science, and computational biology, where DNA, RNA, and protein sequences are represented as character chains. Identifying structural regularities in strings is essential for efficient data processing, pattern recognition, and compression. Among these, \emph{periods}~\cite{Thue06} and \emph{quasiperiodicities} (\emph{covers})~\cite{ApostolicoE90,ApostolicoFI91,MhaskarS22} are widely studied in combinatorial pattern matching, with applications in text indexing, bioinformatics, and data compression.

%Understanding the structural properties of strings and identifying regularities within them is crucial for efficient data processing, pattern recognition, and compression. These notions are widely explored in combinatorial pattern matching and have applications in text indexing, bioinformatics, and data compression.

In recent years, \emph{Characters-Distance-Sampling} (\texttt{CDS}) has been successfully applied in string processing tasks such as exact pattern matching~\cite{FaroM20,FaroMP21,FaroMP23}. \texttt{CDS} represents a string by storing the distances between occurrences of selected \emph{pivot characters}, typically leading to a more compact representation of the string. However, despite its effectiveness in searching and indexing, \texttt{CDS} has not yet been leveraged for computing string regularities.

In this paper, we present the first application of \texttt{CDS} to the computation of fundamental string regularities, specifically the \emph{period} and the \emph{shortest cover} of a string. We develop efficient algorithms that operate directly on the \texttt{CDS} representation of a string, significantly improving computational efficiency compared to classical methods. Our experimental results demonstrate that these methods achieve substantial speedups over traditional approaches.

%The remainder of this paper is organized as follows. Section~\ref{sect-pre} introduces fundamental concepts and notations. Section~\ref{sect-period} presents our method for computing the period of a string using its \texttt{CDS} representation. In Section~\ref{sect-cover}, we describe how to compute the shortest cover of a string using the same representation. Experimental results are detailed in Section~\ref{sect-expe}, followed by our conclusions and future research directions in Section~\ref{sect-conc}.

\section{\label{sect-pre}Preliminaries}

A \emph{string} is a sequence of zero or more characters drawn from a finite \emph{alphabet} $\Sigma$, with $\sigma = |\Sigma|$ denoting the number of distinct characters. The empty string, containing zero characters, is denoted by $\varepsilon$, and the set of all finite strings over $\Sigma$ is $\Sigma^*$. A string $x$ of length $m$ is represented as $x[0\pp m-1]$, where $x[i] \in \Sigma$ for $0 \leq i \leq m-1$, where $|x|=m$ is the length of the string. 

A string $u$ is a \emph{prefix} of $x$ if $x = uw$ for some $w \in \Sigma^*$, and a \emph{suffix} if $x = wu$. If $u$ is both a prefix and a suffix of $x$ but $u \neq x$, it is called a \emph{border}. By definition, the empty string is a border of any non-empty string. A string may have multiple borders, with the longest referred to as \textit{the border} of $x$, denoted by $\Border(x)$. 

An integer $p$ ($1 \leq p \leq |x|$) is a \emph{period} of $x$ if $x[i] = x[i + p]$ for all $0\leq i < |x| - p$. The smallest such $p$ is called \emph{the period} of $x$, denoted by $\per(x)$. It is well known that $x$ has period $p$ if and only if it has a border of length $|x| - p$~\cite{SmythY09}. Specifically, for a non-empty string $x$:
$
|x| = \per(x) + |\Border(x)|.
$

Borders and periods are dual concepts since each period corresponds to a border, and vice versa. All borders of a non-empty string $x$ can be found recursively through the borders of its longest border. In online string matching, a string’s period determines the optimal shift length when an occurrence is found.

\begin{comment}
An alternative definition of periodicity uses occurrences of a smaller string $u$ within $x$. Let $\pos(u, x)$ be the ordered list of all starting positions where $u$ appears in $x$:
$
\pos(u, x) = (i \mid x[i\pp i+|u|-1] = u).
$

The set of distances between consecutive occurrences of $u$ is:
$
\gap(u, x) = \{\pos(u, x)[k+1] - \pos(u, x)[k] \mid 0\leq k < |\pos(u, x)| - 1\}.
$

Then $|u|$ is the period of $x$ if $0 \in \pos(u, x)$, all elements of $\gap(u, x)$ are equal to $|u|$, and the remaining suffix after the last occurrence of $u$ is a prefix of $x$. This guarantees that every character of $x$ is covered by an occurrence of $u$, except possibly for the last few characters, which should align with a prefix of $u$.
\end{comment}

A related concept is that of \emph{covers}. A string $w$ of length $\ell$ is a \emph{cover} (or \emph{quasiperiod}) of a string $x$ of length $m$ if every position in $x$ is included in some occurrence of $w$.
Let $\pos(u, x)$ be the ordered list of all starting positions where $u$ appears in $x$:
$
\pos(u, x) = (i \mid x[i\pp i+|u|-1] = u)
$. Then
$
\{0, |x| - |w|\} \subseteq \pos(w, x)
$
if $w$ is a cover of $x$.

This ensures that every character of $x$ is contained within at least one occurrence of $w$, though occurrences may overlap. A string can have multiple covers, with the shortest among them called the \emph{shortest cover}. If a string $x$ has no proper cover (i.e., its only cover is itself), it is called \textit{superprimitive}; otherwise, it is \textit{quasiperiodic}.

\section{Characters-Distance-Sampling in Brief}\label{sec:cds}

In this section, we provide a concise overview of the methodology used to construct a partial index in the \texttt{CDS} setting.

Let $x$ be a string of length $m$ over an alphabet $\Sigma$ of size $\sigma$. Assuming that all strings are treated as vectors starting at position $0$, we denote $x[i]$ as the $(i+1)$-th character of $x$, for $0 \leq i < m$.

To build the partial index, a subset $C \subseteq \Sigma$ is chosen as the \emph{set of pivot characters}. The input string $x$ is then sampled by computing the distances between consecutive occurrences of any pivot character $c \in C$. 

Formally, given the function $\delta:\{0, \dots, \dot{m}-1\} \rightarrow \{0, \ldots, m-1\}$, where $\delta(i)$ is the position of the $(i+1)$-th occurrence a pivot character in $x$, we first obtain the \emph{position-sampled} version of $x$, indicated by $\dot{x}$, as the numeric sequence of length $\dot{m}$, defined as 
$	\dot{x} = \langle \delta(0), \delta(1), \ldots, \delta(\dot{m}-1) \rangle$.

Then, given the function $\Delta$, where $\Delta(i) = \delta(i+1)-\delta(i)$ for $0 \leq i < \dot{m}-1$, is the distance between two consecutive occurrences of any pivot character in $x$, we get the \emph{character-distance sampled} version of the text $x$, indicated by $\bar{x}$, as the numeric sequence of length $\bar{m}=\dot{m}-1$, defined as 
$\bar{x} = \langle \Delta(0), \Delta(1), \ldots, \Delta(\bar{m}-1) \rangle = \langle \delta(1)-\delta(0), \delta(2)-\delta(1), \ldots, \delta(\dot{m}-1)-\delta(\dot{m}-2) \rangle$.

For example, let $x =$ \texttt{agaacgcagtata} be a \textsc{DNA} sequence of length 13, over the alphabet $\Sigma=\{$\texttt{a}, \texttt{c}, \texttt{g}, \texttt{t}$\}$. Let $C =\{$\texttt{a}$\}$ be the set of pivot characters. Thus, the position-sampled version of $x$ is $\dot{x} = \langle0,2,3,7,10,12\rangle$. The character-distance sampling version of $x$ is $\bar{x} = \langle 2,1,4,3,2 \rangle$. Specifically, $\bar{x}[0]=\Delta(0)=\delta(1)-\delta(0)=2-0=2$, while $\bar{x}[2]=\Delta(2)=\delta(3)-\delta(2)=7-3=4$, and so on.

When searching for a pattern $p$, in a larger text $t$, a preprocessing step computes their sampled versions $\bar{p}$ and $\bar{t}$, after which any string matching algorithm can be applied to find occurrences of $\bar{p}$ in $\bar{t}$ (see \cite{FaroMP20} for details). 
In its original design, the algorithm maintains the position-sampled version $\dot{t}$ and dynamically computes $\bar{t}$ during the search. This approach balances memory usage, as storing distances or positions requires the same space, but retaining positions allows direct access to the original text for verification.

However, in this formulation, each element of the sampled text requires 4 bytes of memory. An alternative implementation~\cite{fake24} reduces memory usage by shifting from position-based to distance-based representation, introducing the concept of \emph{fake samples}. Here, distances exceeding a fixed threshold are decomposed into a sequence of smaller distances with inserted fake samples. This method reduces memory consumption, allowing representation with just 8 bits per element (for a threshold of 256).

Experimental results indicate that the additional space required to store the sampled text ranges from $2.5\%$ to $20\%$, depending on the pattern length~\cite{fake24}. Furthermore, recent studies demonstrate that the Character Distance Sampling algorithm enables any exact string matching procedure to be used as the search method~\cite{FaroMM24}. Notably, search times are reduced by up to $96\%$ compared to applying the same procedure to the non-sampled version of the text.
\subsection{Classical Computation of the Period of a String}

A fundamental approach for computing the period of a string $x$ relies on determining the length of the border of all its prefixes. The key observation is that the length of the longest border of $x$ directly determines its period.

The \emph{border array} of a non-empty string $x$ of length $m$ is an array of size $m+1$ defined as follows:

$$
\border_x[0] = -1, \quad \border_x[i] = |\Border(x[0\pp i-1])| \quad \text{for } 1\le i \le |x|.
$$

The border array can be computed in linear time using the classical algorithm presented in~\cite{chl2007}.
This structure allows the period of the string to be determined efficiently using the relation $\per(x) = m - \border_x[m]$, where $m$ is the length of $x$. 
%\figurename{}~\ref{fig:algo-borders} 

\begin{comment}
\begin{figure}
    \begin{center}
\begin{algo}{Borders}{x,m}
\SET{\border_x[0]}{-1}
\SET{i}{0}
\DOFORI{j}{1}{m-1}
    \SET{\border_x[j]}{i}
    \DOWHILE{i> -1 \mbox{ and } x[i]\neq x[j]}
        \SET{i}{\border_x[i]}
    \OD 
    \SET{i}{i+1}
\OD
\SET{\border_x[m]}{i}
\RETURN{\border_x}
\end{algo}
    \end{center}
    \caption{\label{fig:algo-borders}
    Classical computation of the border array of string $x$ of length $m$.
    }
\end{figure}
\end{comment}

\figurename{}~\ref{fig:example-borders} illustrates the computation of the border array for the string $x = \sa{abaababaaba}$. The border array $\border_x$ encodes structural information about the string by tracking the longest proper border of each prefix of $x$. 

\begin{comment}
For this specific string, the border array values are given by:

\[
\border_x = (-1, 0, 0, 1, 1, 2, 3, 2, 3, 4, 5, 6).
\]
\end{comment}

Each entry $\border_x[i]$ represents the length of the longest border for the prefix $x[0 \pp i-1]$. This means that at position $i$, the prefix $x[0\pp i-1]$ contains a repeated pattern of length $\border_x[i]$. 

For instance, the prefix $x[0 \pp 4] = \sa{abaab}$ has a border of length $2$, meaning that its first two characters match its last two characters. Similarly, the prefix $x[0 \pp 10] = \sa{abaababaaba}$ has a border of length $6$, indicating that the first six characters match the last six. 

\begin{figure}[H]
    \begin{center}
\begin{tabular}{@{}l@{\quad}*{13}{p{9pt}}l@{}}
$i$             &0&1&2&3&4&5&6&7&8&9&10&11 \\
\hline
$x[i]$  &\sa a&\sa b&\sa a&\sa a&\sa b&\sa a&\sa b&\sa a&\sa a&\sa b&\sa a\\
\hline
$\border_x[i]$    &-1&0&0&1&1&2&3&2&3&4&5&6  \\  
\end{tabular}
    \end{center}
    \caption{\label{fig:example-borders}
    Border array of $x=\sa{abaababaaba}$ of length $11$.
    }
\end{figure}

%\figurename{}~\ref{fig:example-borders} shows how the period of the string $x = \sa{abaababaaba}$ is computed. The border array $\border_x$ provides crucial information, where each entry $\border_x[i]$ represents the length of the longest proper border of the prefix $x[0\pp i-1]$. 

For the full string of length $m = 11$, the border array value at index $m$ is given by $\border_x[11] = 6$. Since the period of a string is calculated as $m - \border_x[m]$, we derive:
$
\per(x) = 11 - \border_x[11] = 11 - 6 = 5.
$

This result confirms that $x$ repeats itself every $5$ characters, meaning the shortest repeating unit in $x$ has length $5$.

%%%%%%%%%%%%%%%%%%%%%%%%%%%%%%%%%%%%%%%%%%%%%%%%%%%%%%%%
\subsection{Classical Computation of the Shortest Cover of a String}

Since a cover is inherently a border, computing the shortest cover of a string requires first identifying all its non-empty borders. For a string $x$ of length $m$ let us define the sequence of borders recursively as follows:

$$
\Border^0(x) = x, \quad \Border^i(x) = \Border(\Border^{i-1}(x)) \text{ for } i > 0.
$$
The sequence of non-empty borders is then:

$$
\Border^1(x), \Border^2(x), \dots, \Border^k(x) \ne \varepsilon, \quad \text{such that } \Border^{k+1}(x) = \varepsilon.
$$
The shortest cover of $x$ is the smallest non-empty border $\Border^j(x)$ that covers all the borders $\Border^g(x)$ for $0 \le g < j$.

The algorithm starts by computing the shortest border and iteratively checks whether each border covers longer ones. If a border $\Border^j(x)$ fails to cover some $\Border^h(x)$ with $h < j$, the process resumes with $j = h$. To ensure a linear-time complexity, when a border $w$ of a prefix $x[0\pp n-1]$ has length $\ell > m/2$, its length is adjusted as:
$
\ell = n - \ell + m \bmod (n - \ell).
$
This effectively reduces $w$ to $uu'$, where $w = u^r u'$ and $u'$ is a prefix of $u$.

\begin{comment}
    
\begin{figure}[!t]
    \begin{center}
\begin{tabular}{@{}l@{\quad}*{13}{p{9pt}}l@{}}
$i$             &0&1&2&3&4&5&6&7&8&9&10&11 \\
\hline
$x[i]$  &\sa a&\sa b&\sa a&\sa a&\sa b&\sa a&\sa b&\sa a&\sa a&\sa b&\sa a\\
\hline
$\border_x[i]$    &-1&0&0&1&1&2&3&2&3&4&5&6  \\  
\end{tabular}
    \end{center}
    \caption{\label{fig:example-covers}
    Computation of the shortest cover of $x=\sa{abaababaaba}$.
    }
\end{figure}
\end{comment}

%From \figurename{}~\ref{fig:example-covers}, we analyze how the shortest cover of the string $x = \sa{abaababaaba}$ is determined. A cover of a string is a repeating pattern that occurs frequently enough to cover every position in the string. 

For instance, to compute the shortest cover, we first extract all the non-empty borders of $x$. The decreasing sequence of the non-empty border lengths for $x[0 \pp 11] = \sa{abaababaaba}$  is:

$$
(\border_x^1[|x|] = 6, \border_x^2[|x|] = 3, \border_x^3[|x|] = 1).
$$

This means that sequence of borders in increasing length of $x$ is: $(\sa{a}, \sa{aba}, \sa{abaaba}).$

For a string to be a cover, every occurrence of the border must occur in a way that ensures complete coverage of $x$. Here, the smallest border, $\sa{a}$, is not sufficient to cover $\sa{aba}$. However, the border $\sa{aba}$ does cover both $\sa{abaaba}$ and the full string $x$. 
Thus, the shortest cover of $x$ is: $\sa{aba}.$
This result confirms that the minimal repeating unit that effectively covers $x$ is of length $3$.

\section{Computing the period of a string from its \texttt{CDS} representation}
\label{sect-period}

This section presents an efficient method for computing the period of a string using its \texttt{CDS} representation. We first consider the case where the alphabet consists of only two distinct characters, meaning $\sigma = |\Sigma| = 2$, thus $\Sigma=\{a,b\}$.
Let $a$ be the only pivot character.
Let $x$ be a string of length $m$. The function $\delta$ denotes the position of occurrences of a designated pivot character in $x$. The \texttt{CDS} representation of $x$ is denoted by $\bar{x}$, and its length is given by $\bar{m} = |\bar{x}|$.

The key idea is that the periodic structure of $\bar{x}$ provides enough information to determine the period of $x$. Before proving this, we present two auxiliary results that establish the fundamental relationships between $x$ and its \texttt{CDS} representation.

From any position $\delta(i)$ of the pivot character $a$, all positions 
 in $\delta$ greater than $\delta(i)$ correspond to $\delta(i)$ plus a sum of elements from $\bar{x}$. These positions also contain $a$.

\begin{lemma}\label{lemma-sum-right}
Let $x[\delta(i)] = a$ for $0 \leq i \leq \bar{m}-1$. Then:
\[
x[\delta(i) + \sum_{j=i}^{k} \bar{x}[j]] = a, \quad \text{for } i \leq k \leq \bar{m}-1.
\]
\end{lemma}

\begin{proof}
Expanding the sum of distances between consecutive pivot occurrences, we obtain:
\[
x[\delta(i) + \sum_{j=i}^{k} \bar{x}[j]] = x[\delta(i) + (\delta(i+1) - \delta(i)) + \dots + (\delta(k+1) - \delta(k))].
\]
Since $\delta(k+1)$ corresponds to the position of the pivot character, it follows that:
\[
x[\delta(k+1)] = a.
\]
\end{proof}

A similar result holds for positions preceding a pivot occurrence.

\begin{lemma}\label{lemma-sum-left}
Let $x[\delta(i)] = a$ for $0 \leq i \leq \bar{m}-1$. Then:
\[
x[\delta(i) - \sum_{j=k}^{i-1} \bar{x}[j]] = a, \quad \text{for } 0 \leq k \leq i-1.
\]
\end{lemma}

\begin{proof}
By summing the distances between previous occurrences of the pivot character, we obtain:
\[
x[\delta(i) - (\delta(k+1) - \delta(k) + \dots + \delta(i) - \delta(i+1))] = x[\delta(k)].
\]
Since $\delta(k)$ corresponds to an occurrence of the pivot character, we conclude that:
$
x[\delta(k)] = a.
$
\end{proof}

\begin{obs}\label{obs1}
Let $k = \sum_{i=0}^{\per(\bar{x})-1} \bar{x}[i]$. By definition of $\per(\bar{x})$, it follows that:
\[
\sum_{i=0}^{\per(\bar{x})-1} \bar{x}[j + i] = k, \quad \text{for } 1 \leq j \leq \bar{m} - \per(\bar{x}).
\]
\end{obs}

Let $x[0]$ be the pivot character. If the first and last character of $x$ are the same ($x[0]=x[m-1]$), the period is determined directly from the sum of the first $\per(\bar{x})$ elements of $\bar{x}$. If they are different ($x[0]\ne x[m-1]$), additional conditions must be checked.

\begin{figure}[H]
    \begin{center}
\begin{tabular}{@{}l@{\quad}*{13}{p{9pt}}l@{}}
$i$             &0&1&2&3&4&5&6&7&8&9&10\\
\hline
$x[i]$  &\sa a&\sa b&\sa a&\sa a&\sa b&\sa a&\sa b&\sa a&\sa a&\sa b&\sa a\\
\hline
$\bar{x}[i]$    &2&1&2&2&1&2\\
\hline
$\border_{\bar{x}}[i]$    &-1&0&0&1&1&2&3\\
\end{tabular}
    \end{center}
    \caption{\label{fig:example-period-cds}
    Border array of the \texttt{CDS} representation of $x=\sa{abaababaaba}$ of length $11$ with pivot \sa{a}.
    Then $\per(\bar{x}) = 6-\border_{\bar{x}}[6]=3$, thus $\per(x) = \bar{x}[0]+\bar{x}[1]+\bar{x}[2]=2+1+2=5$.
    }
\end{figure}

The table in Figure~\ref{fig:example-period-cds} illustrates the computation of the period for $x=\sa{abaababaaba}$. The corresponding \texttt{CDS} representation, denoted by $\bar{x}$, consists of the distances between consecutive occurrences of the pivot character \sa{a}. The border array of $\bar{x}$ is computed, and from it, the period of $\bar{x}$ is determined as $\per(\bar{x}) = 3$. Using this result, we obtain the period of $x$ by summing the first three elements of $\bar{x}$, yielding $\per(x) = 2+1+2 = 5$. This example demonstrates how the periodic properties of $\bar{x}$ provide a direct way to compute the period of $x$ efficiently.

If $x[0] \ne x[m-1]$, the structure of the borders must be analyzed further. Let $x=auav$, where $u \in \Sigma^*$ and $v=b^k$ does not contain any occurrences of $a$. The border of $x$ can be computed by examining the borders of $aua$.

\begin{lemma}\label{lemma-border}
The border of $auav$ is given by $u'v'$, where $u'=\Border(aua)$ is the longest border of $aua$ followed by $v'=b^{k'}$ such that $k\le k'$.
\end{lemma}

\begin{proof}
Assume by contradiction that $w$ is a border of $x = auav$ and that $|w| > |u'v'|$. In this case, $w$ must be of the form $aw'av'$, where $aw'a$ is a border of $aua$ that is longer than $\Border(aua)$. This contradicts the assumption that $\Border(aua)$ is the longest border of $aua$, proving the claim.
\end{proof}

\figurename{}~\ref{fig:example-period-cds2} illustrates the \texttt{CDS} representation of the string $x=\sa{abbababbabb}$ and the steps involved in computing its period. The figure shows the border array of the \texttt{CDS} representation, which allows us to determine the period efficiently. By analyzing the values in the border array and summing the appropriate distances from the \texttt{CDS} representation, we obtain $\per(x) = \bar{x}[0] + \bar{x}[1] + \bar{x}[2] = 3 + 2 + 3 = 8$. This confirms that the period of $x$ is 8, as derived from its compact \texttt{CDS} encoding.

%Algorithm \Call{periodCDSbinary}{x,m} presented in Figure~\ref{fig:algo-periodCDSbinary} implements this strategy. 
When $x$ is of the form $auab^k$, checking if a border of $aua$ is followed by $b^{k'}$ with $k \leq k'$ can be determined in constant time by examining values in $\bar{x}$. Thus the worst-case complexity of the algorithm implementing this method is the same as the worst-case time complexity of the algorithm implementing the classical method. Thus it is linear in the length of $\bar{x}$, which is at most $m$.

\begin{figure}[!t]
    \begin{center}
\begin{tabular}{@{}l@{\quad}*{13}{p{9pt}}l@{}}
$i$             &0&1&2&3&4&5&6&7&8&9&10\\
\hline
$x[i]$  &\sa a&\sa b&\sa b&\sa a&\sa b&\sa a&\sa b&\sa b&\sa a&\sa b&\sa b\\
\hline
$\bar{x}[i]$    &3&2&3\\
\hline
$\border_{\bar{x}}[i]$    &-1&0&0&1\\
\end{tabular}
    \end{center}
    \caption{\label{fig:example-period-cds2}
    Border array of the \texttt{CDS} representation of $x=\sa{abbababbabb}=\sa{abbababbab}^2$ of length $11$ with pivot \sa{a}
    thus
    $k = 2$.
    $\border_{\bar{x}}[3]=1$ but $\bar{x}[\border_{\bar{x}}[3]=1]=2 \leq k=2$.
    $\border_{\bar{x}}[1]=0$ and $\bar{x}[\border_{\bar{x}}[1]=0]=3 > k=2$.
    Since we use the border of $\bar{x}$ of length $\border_{\bar{x}}[1]=0$, we use the period of $\bar{x}$ which is equal to $|\bar{x}|-\border_{\bar{x}}[1]=3$, then we sum the first $3$ elements of $\bar{x}$ to get $\per(x) = \bar{x}[0]+\bar{x}[1]+\bar{x}[2]=3+2+3=8$.
    }
\end{figure}

\begin{figure}[!t]
    \centering
    \begin{minipage}[t]{0.55\textwidth} % Align at the top
        \vspace{0pt} % Ensure alignment at the top
        \begin{algo}{periodCDS}{x,m}
            \SET{\bar{x}}{\text{CDS version of $x$}}
            \SET{\bar{m}}{|\bar{x}|}
            \SET{k}{\text{length of the longest suffix of $x$}
            \EXT{\text{that does not contain any $a$}}}
            \SET{bor}{\border_{\bar{x}}}
            \SET{\dot{x}[0]}{\bar{x}[0]}
            \DOFORI{i}{1}{\mxbar}
                \SET{\dot{x}[i]}{\dot{x}[i-1] + \bar{x}[i]}
            \OD
            \SET{(b,\bar{b})}{\Call{borderCDS}{x,m, \bar{x}, \bar{m},k, \cumsum, bor}}
            \SET{\period}{m-b}
            \RETURN{\period}
        \end{algo}
    \end{minipage}
    \hfill
    \begin{minipage}[t]{0.44\textwidth} % Align at the top
        \vspace{0pt} % Ensure alignment at the top
        \begin{algo}{borderCDS}{x,m,\bar{x},\bar{m},\suflen,\cumsum,bor}
            \SET{i}{\bar{m}}
            \DOWHILE{i>-1}
                \SET{i}{bor[i]}
                \DOWHILE{i>-1 \mbox{ and }\bar{x}[i] \le \suflen}
                    \SET{i}{bor[i]}
                \OD
                \IF{i > -1}
                    \SET{\bar{p}}{\bar{m}-i}
                    \SET{p}{\cumsum[\bar{p}-1]}
                    \SET{b}{m-p}
                    \label{line-quadra}
                    \IF{x[0\pp b-1]= \EXT{x[m-b\pp m-1]}}
                        \RETURN{(b,i)}
                    \FI
                \FI
                \RETURN{(0,-1)}
            \OD
        \end{algo}
    \end{minipage}
    \caption{
        \label{fig:algo-periodCDSbinary}
        Computation of the period of a string $x$ of length $m$ with its \texttt{CDS} representation $\bar{x}$ of length $\bar{m}$  (left) computed using the border array
        of $\bar{x}$ and its cumulative sum array $\cumsum$ (right).
        The variable $k$ corresponds to the number of non-pivot symbols after the rightmost pivot symbol in $x$. The algorithm borderCDS returns  the length of the border of $x$ and the length of the border of $\bar{x}$ that corresponds to the period of $\bar{x}$ that can enable to compute the period of $x$.
    }
\end{figure}

Let us now consider the case where the alphabet contains more than two distinct characters. Suppose that $x$ is of the form $x = auav$, where $u \in \Sigma^*$ and $v$ is a substring that contains no occurrences of the pivot character $a$. In this scenario, the longest candidate border of $x$ computed by the previous algorithm must be explicitly verified to determine whether it is a valid border of $x$.

This candidate border should be followed by a substring $v'$, which, like $v$, contains no occurrences of $a$ and is at least as long as $v$, i.e., $|v'| \geq |v|$.
The algorithm then checks whether the candidate border occurs
 both as a prefix and a suffix of $x$.
If this condition holds, the corresponding border length enables to compute a valid period for $x$. If not, the process continues iteratively, examining the next longest candidate border and its associated suffix $v'$, until either a valid period is found or no further non-empty borders exist.

If no valid period is identified through this process, it follows that $x$ is borderless, and its period is simply $m$, the length of the entire string.

The algorithm presented in \figurename{}~\ref{fig:algo-periodCDSbinary} efficiently implements this method for
 computing the length of the longest border of $x$.
Since several prefix sums of $\bar{x}$ may be required in the computation, the algorithm uses a precomputed cumulative sum array over $\bar{x}$.
This preprocessing step ensures that the necessary sums can be retrieved in constant time, optimizing the overall efficiency of the method.
As in the case of exact string matching with sampling,
 the method only computes border candidates that need to be
 naively checked (line~\ref{line-quadra})
 leading to an overall quadratic
 worst case time complexity.
However, experimental results show that this method is faster than the classical one.

A slight modification of the approach  (see~\figurename{}~\ref{fig:algo-bordersCDS})
 allows it to be extended to compute not only the length of the longest border of $x$ but also all the length of its non-empty non-periodic border.
The algorithm can enable to systematically compute and store these border lengths in decreasing order.

\begin{figure}[!t]
    \begin{center}
\begin{algo}{bordersCDS}{x,m}
  \SET{B}{\emptyset}
  \SET{\bar{x}}{\text{CDS version of $x$}}
    \SET{\bar{m}}{|\bar{x}|}
    \SET{k}{\text{length of the longest suffix of $x$}
    \EXT{\text{that does not contain any $a$}}}
  \SET{\cumsum[0]}{\bar{x}[0]}
  \DOFORI{i}{0}{\mxbar-1}
    \SET{\cumsum[i]}{\cumsum[i-1]+\bar{x}[i]}\label{linecumsum}
  \OD
  \SET{bor}{\border_{\bar{x}}}
  \SET{i}{\mxbar}
  \SET{M}{m}
  \DOWHILE{i > -1}
    \SET{(b,i)}{\Call{borderCDS}{x,m,\bar{x},\bar{m},k,\cumsum,\border_{\bar{x}}}}
    \IF{b > M/2}
      \SET{b}{M-b+M\bmod{M-i}}
    \FI
    \DOWHILE{i > -1 \mbox{ and } \cumsum[i] \ge b}
      \SET{i}{i-1}
    \OD
    \SET{\bar{m}}{i+1}
    \SET{m}{i}
    \IF{b > 0}
      \SET{B}{B\cup\{b\}}
    \FI
  \OD
  \RETURN{B}
\end{algo}
    \end{center}
    \caption{
    \label{fig:algo-bordersCDS}
    Computation of the non-empty non-periodic borders of string $x$ of length $m$ using its \texttt{CDS} representation. This implementation uses the algorithm \textsc{borderCDS} described in Fig.~\ref{fig:algo-periodCDSbinary}.
    }
\end{figure}

\section{Computing the shortest cover of a string from its \texttt{CDS} representation}
\label{sect-cover}

The classical computation of the shortest cover of a string $x$ relies on its border array. Since borders play a fundamental role in determining covers, the approach used to compute borders in the \texttt{CDS} representation can be adapted to efficiently determine the shortest cover as well. 

%Given that the border length computation algorithm from the previous section derives borders using the border array of the \texttt{CDS} representation, we extend this method to identify the shortest cover.
The key idea remains the same: covers, like periods, are closely tied to the border structure of a string.
The shortest cover of $x$ must be a border that also satisfies the covering condition, ensuring that every position in $x$ is part of at least one occurrence of the cover.

The algorithm consists of merely identifying borders and verifying which of these borders serve as valid cover for $x$. This is accomplished by leveraging the border array of the \texttt{CDS} representation of $x$, efficiently determining the shortest cover without needing to explicitly reconstruct the full string.

By utilizing the properties of the \texttt{CDS} representation, the adapted algorithm ensures that the computation of the shortest cover remains efficient. The method minimizes redundant checks and avoids unnecessary recomputation by systematically examining the borders derived from the \texttt{CDS} representation. In case of binary alphabets this can be done in linear worst-case time complexity, while for larger alphabets this achieves quadratic worst-case time complexity, however experimental results demonstrate that this solution is faster in practice compared to the classical computation.

%\begin{figure}
%    \begin{center}
%\begin{algo}{shortestCoverCDS}{x,m}

%  \SET{(\bar{x},\last)}{\Call{computeCDS}{x,m}}
%  \SET{\border_{\bar{x}}}{\Call{Borders}{\bar{x}, |\bar{x}|}}
%  \SET{S}{\Call{emptyStack}{}}
%  \CALL{Push}{S,|\bar{x}|}
%  \SET{n}{|\bar{x}|}
%  \SET{b}{0}
%    \DOFORI{j}{m-\border_{\bar{x}}[n]}{m}
%        \SET{b}{b+\bar{x}[i]}     
%    \OD
%  \SET{i}{b}
%  \DOWHILE{i \ne 0}
%    \IF{i > m/2}
%      \CALL{Push}{S,n-i+n\bmod (n-i)}
%    \ELSE
%      \CALL{Push}{S,i}
%      \ACT{}
%    \FI
%      \SET{b}{0}
%        \DOFORI{j}{m-\border_{\bar{x}}[n]}{m}
%            \SET{b}{b+\bar{x}[i]}     
%        \OD
%    \SET{n}{i}
%    \SET{i}{b}
%  \OD
%  \SET{b}{\Call{Pop}{S}}
%  \DOWHILE{|S|>0}
%\label{line:pos1cds}    \SET{\listpos}{\pos(x[0\pp b-1],x[0\pp \Call{Top}{S}-1])\cup (m)}
%\label{line:iscovering1cds}    \DOWHILE{|S|> 1 \mbox{ and } \Call{isCovering}{\listpos,b}}
%      \CALL{Pop}{S}
%\label{line:pos2cds}      \SET{\listpos}{\pos(x[0\pp b-1],x[0\pp \Call{Top}{S}-1])\cup (m)}
%    \OD
%\label{line:iscovering2cds}    \IF{\Call{isCovering}{\pos,b}}
%      \RETURN{x[0\pp b-1]}
%    \ELSE
%      \SET{b}{\Call{Pop}{S}}
%    \FI
%  \OD
%  \RETURN{x[0\pp b-1]}
%\end{algo}
%    \end{center}
%    \caption{\label{fig:algo-shortestCovercds}
%    Classical computation of the shortest cover of string $x$ of length $m$, using the \texttt{CDS}.
%    }
%\end{figure}

\section{Experimental results}
\label{sect-expe}

In this section, we present experimental results to evaluate the performance of the algorithms proposed in this paper. Our objective is to assess the computational efficiency of the methods that leverage the \texttt{CDS} representation for computing the period and shortest cover of a string.

The algorithms have been implemented in the \textsf{C} programming language and tested using the \textsc{Smart} tool~\cite{FaroLBMM16}. The experiments were conducted on a MacBook Pro equipped with a 2.7 GHz Intel Core i7 processor, four cores, 16 GB of 2133 MHz LPDDR3 RAM, a 256 KB L2 cache, and an 8 MB L3 cache. The compilation was performed using the \textsf{gcc} compiler with the \texttt{-O3} optimization flag to ensure optimal execution speed. \footnote{The \textsc{Smart} tool is available for download at \url{http://www.dmi.unict.it/~faro/smart/} or \url{https://github.com/smart-tool/smart}.}

To obtain statistically significant results, each algorithm was executed over 10,000 different runs, and the average execution time was recorded. The dataset used in our experiments consists of 100MB of English text extracted from the \emph{Pizza and Chili} corpus~\cite{pizzachili}, which is widely adopted in the field of string processing research. The experiments were performed on input strings of varying lengths, ranging from $10^5$ to $10^7$ characters.

For consistency, the pivot character selected for constructing the \texttt{CDS} representation was always the first character of the input string. The additional memory required to store the \texttt{CDS} representation varied between $3\%$ and $15\%$ of the size of the original string, depending on the structure of the input data.

\begin{figure}[!t]
    \centering
    \begin{minipage}[t]{0.48\textwidth}
        \centering
        \vspace{0pt} % Ensures top alignment
        \begin{tikzpicture}
            \begin{axis}[
                width=\textwidth,
                height=5cm,
                ymajorgrids=true,
                xlabel={Size ($10^n$)},
                ylabel={Time (ms)},
                xtick={4,5,6,7,8},
                xmode=log,
                xticklabels={4,5,6,7,8},
                ymode=log,
                log basis y=10,
                log ticks with fixed point,
                legend style={font=\scriptsize},
                ]
                
                \addplot[thick,blue,smooth] coordinates {
                    (5, 0.07) (6, 0.64) (7, 8.40)
                };
                \addplot[thick,red,smooth] coordinates {
                    (5, 0.04) (6, 0.40) (7, 4.56)
                }; 
            \end{axis}
        \end{tikzpicture}
    \end{minipage}
    \hfill
    \begin{minipage}[t]{0.48\textwidth}
        \centering
        \vspace{0pt} % Ensures top alignment
        \begin{tikzpicture}
            \begin{axis}[
                width=\textwidth,
                height=5cm,
                xmode=log,
                ymajorgrids=true,
                xlabel={Size ($10^n$)},
                xtick={4,5,6,7,8},
                xticklabels={4,5,6,7,8},
                ymode=log,
                log basis y=10,
                log ticks with fixed point,
                legend style={font=\scriptsize},
                ]
                
                \addplot[thick,blue,smooth] coordinates {
                    (5, 0.07) (6, 0.61) (7, 8.32)
                };
                \addplot[thick,red,smooth] coordinates {
                    (5, 0.02) (6, 0.23) (7, 2.84)
                }; 
                \legend{\texttt{STD}, \texttt{CDS}}
            \end{axis}
        \end{tikzpicture}
    \end{minipage}
        \caption{
        (Left)
        Average execution times of the standard shortest cover computation algorithm and its optimized version using the \texttt{CDS} representation over 10,000 runs.\\
        (Right) Average execution times of the standard period computation algorithm and its optimized version using the \texttt{CDS} representation over 10,000 runs.}
        \label{fig:algorithms_performance_period}
\end{figure}

The results in Figure~\ref{fig:algorithms_performance_period} indicate that the method for computing the period using the \texttt{CDS} representation achieves a speedup between $38\%$ and $43\%$ compared to the standard approach. This improvement is attributed to the more compact nature of the \texttt{CDS} representation, which reduces the number of operations required to determine the periodicity of the string.

Similarly, Figure~\ref{fig:algorithms_performance_period} also presents the execution times for computing the shortest cover of a string using both the standard method and the optimized \texttt{CDS}-based approach. The results show that the \texttt{CDS} representation leads to a significant performance gain, achieving speedups between $63\%$ and $72\%$. This improvement is due to the ability of the \texttt{CDS} representation to efficiently encode the positions of relevant characters, allowing the algorithm to determine covers with fewer operations.

These experimental results demonstrate the efficiency of the \texttt{CDS}-based approaches for computing both the period and shortest cover of a string. By leveraging the compactness and structure of the \texttt{CDS} representation, these methods significantly reduce computation time while maintaining correctness. The proposed techniques are particularly well-suited for applications involving large-scale text processing, where reducing computational overhead is critical.

\section{Conclusion}
\label{sect-conc}

In this paper, we demonstrated that the partial index built using the \emph{Characters-Distance-Sampling} (\texttt{CDS}) algorithm, originally designed to enhance the efficiency of string matching algorithms, can also be effectively utilized to compute regularities in strings. This novel approach leverages the compact and structured nature of the \texttt{CDS} representation to extract fundamental properties of strings in a computationally efficient manner.

Specifically, we introduced optimized methods for solving two well-studied problems in the field of string regularities: the computation of the \emph{period} and the detection of \emph{quasiperiodicity} or \emph{covers}. By adapting classical techniques to work with the \texttt{CDS} representation, we developed algorithms that significantly outperform traditional approaches. Experimental results confirm the effectiveness of our methods, showing speedups ranging from $38\%$ to $43\%$ for period computation and from $63\%$ to $72\%$ for shortest cover detection. These improvements demonstrate the potential of the \texttt{CDS} representation not only for accelerating string matching but also for efficiently identifying structural properties of strings.

Our findings open the door for further research into the use of \texttt{CDS}-based methods for detecting other types of regularities in strings. Future work could explore the application of this approach to problems such as the computation of \emph{seeds}, \emph{k-covers}, or other forms of repetitive structures. Investigating the theoretical bounds of this technique, as well as its practical implications in bioinformatics, data compression, and text processing, could lead to new advancements in the field.

The same techniques can be easily adapted to strings given in other kind of representation such as the run-length encoding.

By demonstrating that a representation initially designed for string matching can be repurposed to analyze structural properties of strings, this work provides a new perspective on the versatility of distance-based indexing techniques. The promising results obtained in this study suggest that \texttt{CDS}-based methods could be further developed to enhance the efficiency of a wide range of string-processing applications.

\bibliographystyle{splncs04} 
\bibliography{bibliography.bib}
\end{document}